\documentclass[11pt, onecolumn]{article}
\usepackage{amsmath,amsthm,amssymb}
\usepackage{graphicx}
\usepackage{wrapfig}
\usepackage{caption}
\usepackage{subcaption,enumerate, color}
\usepackage[margin=1in]{geometry}
\usepackage{float}
\usepackage{multirow}
\usepackage[dvips]{epsfig}
\usepackage{lingmacros}
\usepackage{amsfonts}
\usepackage{bbm}
\usepackage{tikz}
\usepackage[utf8]{inputenc}
\usepackage[T1]{fontenc}

\usepackage{pgfplots}
\pgfplotsset{compat=1.12}

\usepackage{authblk}

\usetikzlibrary{graphs}
\usetikzlibrary{graphs.standard}
\usepackage{tkz-berge}
\usetikzlibrary{decorations.markings}
\usepackage{multicol, pgfplots}

\usepackage{tikz}
\usetikzlibrary{arrows,trees}
\usetikzlibrary{shapes,backgrounds,calc,intersections,patterns}

\theoremstyle{definition}
\newtheorem{definition}{Definition}[section]

\newtheorem{theorem}[definition]{Theorem}
\newtheorem{lemma}[definition]{Lemma}

\newtheorem{cor}[definition]{Corollary}

\newcommand{\D}{\mathcal{D}}

\newcommand{\R}{{\mathbb R}}

\newcommand{\rmv}[1]{}

\newcommand{\pP}[1]{\mbox{Prob}\left[ #1 \right] }
\newcommand{\pPP}[2]{\mbox{Prob}_{#2} \left[ #1 \right] }

\newcommand{\VV}[2]{\text{Vol}_{#2}\left( #1 \right)}

\definecolor{OliveGreen}{rgb}{0.0, 0.6, 0.0}

\makeatletter
\g@addto@macro\bfseries{\boldmath}
\makeatother

\DeclareMathOperator{\Jac}{Jac}

\usepackage{todonotes}

\title{Optimal Bounds for Johnson-Lindenstrauss Transformations
}

\newcommand*\samethanks[1][\value{footnote}]{\footnotemark[#1]}

\author[1]{Michael Burr\thanks{Partially supported by a grant from the Simons Foundation (\#282399 to Michael Burr) and National Science Foundation Grant CCF-1527193.}}
\author[1]{Shuhong Gao\thanks{Partially supported by the National Science Foundation  under Grants CCF-1407623, DMS-1403062, and  DMS-1547399.}}
\author[2]{Fiona Knoll\samethanks\thanks{Most of the work was done while the author was at Clemson University.}}
\affil[1]{Department of Mathematical Sciences\\Clemson University, Clemson,  SC 29634}
\affil[2]{Department of Mathematical Sciences\\University of Cincinnati, Cincinnati, OH 45221}

\begin{document}

\maketitle

\begin{abstract}
In 1984, Johnson and Lindenstrauss proved that any finite set of data in a high-dimensional space can be projected to a lower-dimensional space while preserving the pairwise Euclidean distance between points up to a bounded relative error.
If the desired dimension of the image is too small, however,  
Kane, Meka, and Nelson (2011) and Jayram and Woodruff (2013) independently proved that such a projection does not exist.
In this paper, we provide a precise asymptotic threshold for the dimension of the image, above which, there exists a projection preserving the Euclidean distance, but, below which, there does not exist such a projection.
\end{abstract}



\section{Introduction}

In 1984, Johnson and Lindenstrauss \cite{JohnsonLindenstrauss}, in establishing  a bound on the Lipschitz constant for the Lipschitz extension problem, proved that 
any finite set of data in a high-dimensional space can be projected into a lower-dimensional space while preserving the pairwise Euclidean distance within any desired relative error.
In particular,  for any finite set of vectors $x_1, \ldots, x_N \in \R^d$ and for any error factor $0 < \epsilon < \frac{1}{2}$, 
there exists an absolute  constant $c$ such that for all $k \geq c \epsilon^{-2} \log N$, 
there exists a linear map $A: \R^d \rightarrow \R^k$  such that for all pairs $1\leq i,j \leq N$,
		\[(1-\epsilon) \|x_i-x_j\|_2 \leq \|Ax_i - Ax_j\|_2 \leq (1+\epsilon) \|x_i-x_j\|_2, \]
where $\|\cdot \|_2$ denotes the Euclidean norm.  These inequalities are implied by the following theorem
(by setting $\delta = \frac{1}{N^2}$ and taking the union bound):
\begin{theorem}[Johnson and Lindenstrauss \cite{JohnsonLindenstrauss}]\label{Thm: Main JL 2}
For any real numbers  $0< \epsilon, \delta <\frac{1}{2}$, there exists an absolute constant $c>0$ such that for  
any integer  $k \geq c \epsilon^{-2} \log \frac{1}{\delta}$, there exists a 
probability distribution $\D$ on $k \times d$ real matrices such that for any fixed $x \in \R^d$,
				\begin{equation} \label{Ineq: JL}
					\pPP{ (1-\epsilon)\|x\|_2^2 \leq \|Ax\|_2^2   \leq (1+\epsilon) \|x\|_2^2}{A \sim \D} > 1-\delta, 
				\end{equation}	
where $A\sim \D$ means that the  matrix $A$ is a random matrix with distribution $\D$.
\end{theorem}			
Note that, in order to  project  a large number of vectors, $\delta$ must be sufficiently small. For instance, suppose we wish to project a set of 
$N=2^{20}$ vectors to a smaller dimensional space.  To apply the union bound to Inequality (\ref{Ineq: JL}), we use $\delta=2^{-40}$.  In this case, Inequality (\ref{Ineq: JL})  implies that  the probability of preserving all pairwise distances between $N$ points
(up to a relative error of $\epsilon$) is at least  $1-\delta N^2/2 = 1/2$.  Since the probability is nonzero, such a projection exists.

A  probability distribution $\D$ satisfying  Inequality (\ref{Ineq: JL}) is called an    $(\epsilon,\delta)$-JL distribution,  or simply a JL distribution.  Since these transformations are linear, without loss of  generality,  
we assume for the rest of the paper that $\|x\|_2=1$.  When a JL-distribution is specified via an explicit construction, we may call a random projection $x\mapsto Ax$ generated in this way a JL transformation.
		
Since the introduction of JL distributions, there has been considerable work on explicit constructions of JL distributions, see, e.g.,
 \cite{JohnsonLindenstrauss, FranklMaehara, IndykMotwani, Achlioptas, AilonChazelle_ANN, Matousek, DasguptaKumarSarlos, KaneNelson_Sparser} and the references therein. 
 A simple and easily described  JL distribution is that of Achlioptas \cite{Achlioptas}. In this construction, the entries  of $A$ are distributed as follows:
 \[a_{ij} = \sqrt{\frac{3}{k}}\cdot \left\lbrace 
		 \begin{matrix} 1, & \text{ with probability } 1/6, \\ 
			               0, & \text{ with probability } 1/3, \\  
			            - 1, & \text{ with probability } 1/6.	
               \end{matrix} \right.
\]
The recent constructions in \cite{AilonChazelle_ANN, Matousek, DasguptaKumarSarlos, KaneNelson_Sparser} have focused on the complexity of computing the 
 projection for the purpose of applications. We note that the  ability to project a vector to a smaller dimensional space, independent of the original dimension, while preserving the Euclidean norm up to a prescribed relative error, is highly desirable.  In particular, dimension reduction has applications to many fields, including machine learning 
\cite{MachineLearning_Vempala, MachineLearning_Weinberger}, low rank approximation  \cite{LowRank_ClarksonWoodruff, LowRank_Nguyen, LowRank_Ubaru}, approximate nearest neighbors \cite{AilonChazelle_ANN, IndykMotwani}, data storage \cite{RIP_Candes, Streaming}, and document similarity \cite{DocSim_Bingham, DocSim_Lin}.

For both practical and theoretical purposes,  it is important to know the smallest possible dimension $k$ of a potential image space 
for any given  $\epsilon$ and $\delta$.  
Note that, for any $d_1 < d$, each $(\epsilon, \delta)$-JL distribution $\D$ on $\R^{k\times d}$ induces an  $(\epsilon, \delta)$-JL distribution $\D_1$ on $\R^{k\times d_1}$  in a natural way: 
the matrices  of $\D_1$ are  obtained from $\D$ by deleting the last $d-d_1$ columns, together with the induced probability distribution.   This construction is a JL distribution since $\R^{d_1}$ can be naturally embedded into $\R^d$ by extending a vector in $\R^{d_1}$ by $d-d_1$ zeros.
Hence, if there exists an $(\epsilon, \delta)$-JL distribution on $\R^{k\times d}$, then there is an $(\epsilon, \delta)$-JL distribution on $\R^{k\times d_1}$ for all $1 \leq d_1 \leq d$.  Similarly, if an $(\epsilon,\delta)$-JL distribution does not exist on $\mathbb{R}^{k\times d}$, then, for any $k_1<k$, then there cannot be an $(\epsilon,\delta)$-JL distribution on $\mathbb{R}^{k_1\times d}$.  In particular, since $\mathbb{R}^{k_1}$ can be naturally embedded into $\mathbb{R}^k$ by extending a vector in $\mathbb{R}^{k_1}$ by $k-k_1$ zeros, if an $(\epsilon,\delta)$-JL distribution existed for $\mathbb{R}^{k_1\times d}$, it could be extended to an $(\epsilon,\delta)$-JL distribution existed for $\mathbb{R}^{k\times d}$.

For any $\epsilon$ and $\delta$, we define
\[k_0(\epsilon,\delta) = \min \{k:  \mbox{there exists an $(\epsilon, \delta)$-JL distribution on $\R^{k\times d}$ for every  $d\geq 1$} \}.   \]
By our definition,  $k_0= k_0(\epsilon,\delta)$ is independent of $d$, and, by  Theorem \ref{Thm: Main JL 2},  we have $k_0 \leq c \epsilon^{-2} \log(1/\delta)$ for some absolute constant $c>0$.  
 Frankl  and Maehara \cite{FranklMaehara}  show that $c \leq 9$. 
Achlioptas \cite{Achlioptas} further improves  this bound by providing a JL distribution with
			\[k> 2 \log(2/\delta) \left(\frac{\epsilon^2}{2} - \frac{\epsilon^3}{3} \right)^{-1} ,\]
resulting in the following upper bound:
\[k_0  \leq  2 \log(2/\delta) \left(\frac{\epsilon^2}{2} - \frac{\epsilon^3}{3} \right)^{-1} = 4\epsilon^{-2} \log(1/\delta) \left[ 1+ o(1) \right],\]
		where $o(1)$ approaches zero as both $\epsilon$ and $\delta$ approach zero.	
		
A lower bound on $k_0$ was not given until 2003 when Alon \cite{Alon} proved that 
					\[k_0 \geq c \epsilon^{-2} \log(1/\delta) \Big/ \log(1/\epsilon) \] 
for some absolute constant $c>0$.	Improving Alon's work, Jayram and Woodruff \cite{WoodruffJayram} and  Kane, Meka, and Nelson \cite{KaneNelsonMeka}  
showed, through different methods, that, for some absolute constant $c_1>0$,  there is no $(\epsilon,\delta)$-JL distribution for 
$k \leq c_1\epsilon^{-2} \log \frac{1}{\delta}$.  Hence,  there is a lower bound of the form
$k_0 \geq c_1  \epsilon^{-2} \log \frac{1}{\delta}$. This situation is summarized in Figure \ref{Figure:Line}.
		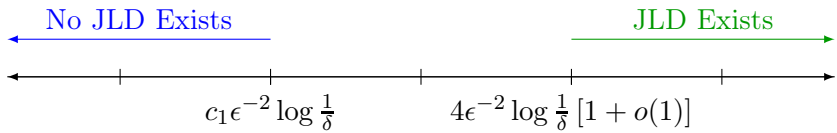
\begin{figure}[H]
		\begin{center}
		\begin{tikzpicture}
		\draw[latex-] (-5.5,0) -- (5.5,0) ;
		\draw[-latex] (-5.5,0) -- (5.5,0) ;
		\foreach \x in  {-4,-2,0,2,4}
		\draw[shift={(\x,0)},color=black] (0pt,3pt) -- (0pt,-3pt);
		
		\draw[shift={(-2,0)},color=black] (0pt,0pt) -- (0pt,-3pt) node[below] 
		{$c_1 \epsilon^{-2} \log \frac{1}{\delta}$};
		\draw[shift={(2,0)},color=black] (0pt,0pt) -- (0pt,-3pt) node[below] 
		{$4 \epsilon^{-2} \log \frac{1}{\delta}\left[ 1+ o(1) \right]$};
		
		\draw[-latex,color=OliveGreen] (2,.5)--(5.5,.5); 
		\draw[color = OliveGreen] (3.75, .5) node[above] {JLD Exists};
		
		\draw[latex-,color=blue] (-5.5,.5)--(-2,.5); 
		\draw[color = blue] (-3.75, .5) node[above] {No JLD Exists};
		
		
		\end{tikzpicture}	
		\caption{For fixed $\epsilon$ and $\delta$, there exists a JL distribution for $k \geq 4 \epsilon^{-2} \log(1/\delta)\left[ 1+ o(1) 
		\right]$. For $k<c_1 \epsilon^{-2} \log(1/\delta)$, for some absolute  constant $c_1>0$,  there is no JL distribution.  In this paper, we close this gap in the limit.
		\label{Figure:Line}}
		\end{center}
\end{figure}	
		
The goal of the current  paper is to close the gap between the upper and lower bounds in the limit.  In particular, we prove an optimal lower bound that  asymptotically matches the known upper bound when $\epsilon$ and $\delta$ approach $0$, see 
Theorem \ref{Thm:IntroThm}.   This means that $4 \epsilon^{-2} \log(1/\delta)$ is an asymptotic threshold for $k_0$ where  a phase change phenomenon occurs. 
\begin{theorem} \label{Thm:IntroThm}
For $\epsilon$ and $\delta$ sufficiently small, $k_0 \approx 4\epsilon^{-2} \log(1/\delta)$.  More precisely,
\[ \lim_{\epsilon, \delta \rightarrow 0} \  \frac{k_0(\epsilon, \delta)}{4\epsilon^{-2} \log(1/\delta)} =  1.\]
\end{theorem}				

The rest of the paper is organized as follows: To prove Theorem \ref{Thm:IntroThm}, we follow the approach of  Kane, Meka and  Nelson \cite{KaneNelsonMeka}. 
To make their constant $c_1$ explicit, however,  we  must use a more careful argument.  
In Section \ref{Sec:Main:New}, we provide explicit conditions under which we prove the main result, Theorem \ref{Thm:IntroThm}.  
We delay the proofs of the explicit conditions until Sections \ref{Sec:UniformMeasure} and \ref{Concentration} in order to make the main result more accessible since only the statements of these results (which are of independent interest) are needed, and not their more technical proofs.
In Section \ref{Sec:UniformMeasure},   we study uniform distributions  and surface areas (or hypervolumes) on high-dimensional spheres. 
More precisely,  for any $d\geq 1$, let $S^{d-1}$ denote the unit sphere of dimension $d-1$, i.e., 
$S^0 = \{1, -1\}$ has two points, $S^1$ is the unit circle, $S^2$ is the unit sphere in $\R^3$, and, in general,
\[ S^{d-1} = \left\lbrace x \in \R^{d}: \sum_{i=1}^d x_i^2 = 1 \right\rbrace,\]
and $d\Omega_{d-1}$ be the surface area measure for $S^{d-1}$. We show that, for any $1\leq k \leq d$, 
  \[ d\Omega_{d-1} = \frac{1}{2} f(s)ds\, d\Omega_{k-1} d\Omega_{d-k-1},\]
where $s \in [0,1]$ and   $f(s) =  s^{\frac{k-2}{2}}(1-s)^{\frac{d-k-2}{2}}$.  
This is a more precise version of a result in \cite{KaneNelsonMeka}, replacing an unspecified constant by $1/2$.
This formula  is of independent interest since it shows that the uniform distribution on $S^{d-1}$ is a product of uniform distributions on $S^{k-1}$ and $S^{d-k-1}$ with a distribution on $[0,1]$, see Theorem \ref{Thm_ud}.  In Section \ref{Concentration}, we prove probabilistic bounds  on $s =x_1^2 + \cdots + x_k^2$ where
 $x=(x_1, \cdots, x_k, \cdots, x_d)$ is a random variable uniformly distributed on $S^{d-1}$.  These bounds can be viewed as explicit bounds for concentration theorems for laws of large numbers in probability theory.  
  
\section{Asymptotic Threshold Bound} \label{Sec:Main:New}

In this section, we prove the asymptotic threshold bound for JL transformations.  In particular, we provide specific conditions that result in the asymptotic threshold bound of $4\epsilon^{-2}\log(1/\delta)$.  In Sections \ref{Sec:UniformMeasure} and \ref{Concentration}, we prove that these specific conditions hold, but the details of these proofs are more technical, and only the statements are needed for the asymptotic bound. 

\subsection{The Uniform Distribution on $S^{d-1}$}

There is a unique probability distribution, called the uniform distribution, on $S^{d-1}$ that is invariant under the orthonormal group.
From a sampling point of view, a uniform random point on $S^{d-1}$ can be obtained as follows:
Let $x_1, x_2, \ldots, x_d$ be independent random variables on $\R$ distributed according to the Gaussian distribution  $N(0,1)$ (i.e., the standard normal distribution with mean $0$ and variance $1$),  and let $X=(x_1,x_2, \ldots, x_d)^t$. 
Then,  $x= \frac{X}{\|X\|_2}$ is a random point uniformly distributed on $S^{d-1}$.    

The uniform distribution may also be defined in terms of the surface area as follows: 
Let $\VV{S^{d-1}}{d-1}$ denote the $(d-1)$-dimensional surface area (or hypervolume) of $S^{d-1}$, and, similarly, let $\VV{V}{d-1}$ denote the surface area of $V$ for any (measurable) subset $V$ of $S^{d-1}$.   For example,
\[\VV{S^0}{0} =2, \quad  \VV{S^1}{1} = 2 \pi, \quad\text{and}\quad \VV{S^{d-1}}{d-1} =  \frac{2 \pi^{d/2}}{\Gamma(d/2)}, \]
where  $\Gamma(z)$ denotes the Gamma function 
\[\Gamma(z) = \int_0^\infty x^{z-1}e^{-x}dx. \]
The probability that a random point $x$ from $S^{d-1}$ drawn from the uniform distribution is in $V$ equals $\VV{V}{d-1}/\VV{S^{d-1}}{d-1}$,
hence the probability is invariant under orthonormal transformations.

 We express the uniform distribution on $S^{d-1}$  in term of the surface area differential form\footnote{In this paper, we suppress the pullback maps on equalities for differential forms since there is a unique (almost) bijective map under consideration in each case.  We leave the details to the interested reader.} 
 $d \Omega_{d-1}$, which means that, for any measurable subset 
$V \subset S^{d-1}$,  the $(d-1)$-dimensional surface area of $V$ is equal to the integral with respect to $d\Omega_{d-1}$, i.e., $\VV{V}{d-1}=\int_V  d \Omega_{d-1}$.  For example,  $d \Omega_{0} = \delta_{1}+\delta_{-1}$ consists of two point measures and 
$d \Omega_{1} =  \frac{dx_1}{x_2} = - \frac{dx_2}{x_1}$
at any point $(x_1,x_2)^t \in S^1$.     Thus, the uniform distribution on $S^{d-1}$  is defined in terms of $d \Omega_{d-1}/\VV{S^{d-1}}{d-1}$, i.e., 
 for any measurable subset  $V \subset S^{d-1}$, 
 \[ \pPP{x \in V}{x \sim S^{d-1}} = \frac{1}{\VV{S^{d-1}}{d-1}}\int_V d \Omega_{d-1}=\frac{\VV{V}{d-1}}{\VV{S^{d-1}}{d-1}}\]
 where $x \sim S^{d-1}$ means that $x$ is a random variable uniformly distributed on $S^{d-1}$. 

As we are interested in reducing a $d$-dimensional vector to a $k$-dimensional vector for $1 \leq  k < d$,  we derive a relationship between the  uniform distribution  on 
 $S^{d-1}$ and the uniform distributions on $S^{k-1}$ and $S^{d-k-1}$.  Following the approach of  Kane, Meka and Nelson \cite{KaneNelsonMeka}, for $1\leq k<d$, we define an injective map 
\[ \Psi:   S^{d-1}   \rightarrow    [0,1] \times S^{k-1} \times S^{d-k-1} \]
 as follows:  
 For any  $x= (x_1,x_2, \ldots, x_d)^t \in S^{d-1}$, we define $s$ in $\Psi(x) = (s, u , v)$ as $s=x_1^2+\dots+x_k^2$.  In the case where $0<s<1$, we define 
  \[  u =  (x_1,  \ldots, x_k)^t/\sqrt{s} \quad\text{and}\quad   v = (x_{k+1}, \ldots, x_d)^t/\sqrt{1-s}.\]
When $s=0$, i.e., $x_1=\dots=x_k=0$, we define $u=(1,0, \ldots, 0)^t$  (or any point in $S^{k-1}$) and $v = (x_{k+1}, \ldots, x_d)^t$.  
 Similarly, for $s=1$, we define $u=(x_1, \ldots, x_k)^t$ and $v=(1,0, \ldots, 0)^t$ (or any point in $S^{d-k-1}$). 
 It is straight-forward to check that $\Psi$ is injective.  In addition, the complement of the image of $\Psi$ is a subset of $\{0,1\} \times S^{k-1} \times S^{d-k-1}$ which has $(d-1)$-dimensional surface area $0$.  Therefore, when necessary, we assume that $s\in(0,1)$.
 
 For $s \in [0,1]$, we define
 \[ f(s) =  s^{\frac{k-2}{2}}(1-s)^{\frac{d-k-2}{2}}.\]
 In Theorem \ref{Thm_ud}, we prove that, via the map $\Psi$, 
  \[ d\Omega_{d-1} = \frac{1}{2} f(s)ds\, d\Omega_{k-1} d\Omega_{d-k-1}.\]
  Equivalently,  in term of probability distributions,
\begin{equation}\label{eq:equalitymeasures}
  \frac{d\Omega_{d-1}}{\VV{S^{d-1}}{d-1}} =  B f(s) ds \frac{d\Omega_{k-1}}{\VV{S^{k-1}}{k-1}} \frac{d\Omega_{d-k-1}}{\VV{S^{d-k-1}}{d-k-1}},
  \end{equation}
 where $B$ is an appropriate scaling constant depending on $d$ and $k$, for more details, see Equation~(\ref{Eqn:B}).  Moreover, in this situation, $Bf(s)$ is a probability distribution on $[0,1]$.  This implies that the uniform distribution on $S^{d-1}$ is a direct product of the distributions on the factors.  In other words, a uniformly distributed random variable $X_{d-1}$ on $S^{d-1}$ can be decomposed into three random variables $\Psi(X_{d-1})=(S,X_{k-1},X_{d-k-1})$ with the following properties:
\begin{enumerate}[(i)]
\item $S$ is a random variable on $[0,1]$ with density function $Bf(s)$,
\item $X_{k-1}$ and $X_{d-k-1}$ are uniformly distributed on $S^{k-1}$ and $S^{d-k-1}$, and
\item The random variables $S$, $X_{k-1}$, and $X_{d-k-1}$ are {\em independent}.
\end{enumerate}
The independence of these three random variables is a key property in our proof as it allows us to study the three spaces independently.

\subsection{Upper Bound: Explicit JL Distribution} \label{Sec: OrthogProj}

We recall that Achlioptas \cite{Achlioptas} proved that 
\[k_0(\epsilon, \delta)   \leq  2 \log(2/\delta) \left(\frac{\epsilon^2}{2} - \frac{\epsilon^3}{3} \right)^{-1} = 4\epsilon^{-2} \log(1/\delta) \left[ 1+ o(1) \right].\]
In this section, we give an alternate proof of this result using the approach and bounds from this paper.

We recall the following construction  by Gupta and  Dasgupta \cite{DasguptaGupta}:   A distribution $\D$ on $k \times d$ matrices is formed by picking a $d\times d$ orthonormal  matrix 
$V = (v_1, \ldots, v_d)^t$ uniformly at random with respect to the Haar measure on orthonormal matrices, 
and then letting $A = \frac{1}{\sqrt{s_0}}(v_1, \ldots, v_k)^t$ where $s_0= k/d$.  From a sampling perspective, $A$ can be constructed by drawing $v_1$ from a uniform distribution on $S^{d-1}$, and then drawing each $v_i$ from a uniform distribution on the $(d-i)$-dimensional sphere perpendicular to $v_1$, $\dots$, $v_{i-1}$.  The following theorem shows that $k_0(\epsilon, \delta)   \leq  4\epsilon^{-2} \log(1/\delta) \left[ 1+ o(1)\right]$, which, in turn, implies that
the limit appearing in Theorem \ref{Thm:IntroThm} (if it exists) is at most $1$:
\begin{theorem} \label{Thm: exists}
Let $0<\epsilon,\delta<\frac{1}{2}$ and $s_0=k/d$.  Suppose that there is some constant $C$ so that
$$
\max\{\pPP{ s < s_0(1-\epsilon)}{x\sim S^{d-1}},\pPP{ s> s_0(1+\epsilon)}{x\sim S^{d-1}} \}\leq Ce^{-\frac{k-2}{4}\epsilon^2\left(1-\frac{2}{3}\epsilon\right)},
$$
where $s$ is defined as in $\Psi(x)=(s,u,v)$.
Then, there exists an $o(1)$ function, which approaches zero as both $\epsilon$ and $\delta$ approach zero so that if $k> 4\epsilon^{-2} \log \left(\frac{1}{\delta} \right) \left[1+o(1)\right]$, then the distribution on $k\times d$ random matrices defined as above is 
an $(\epsilon, \delta)$-JL distribution, 
that is,   for any $w \in S^{d-1}$, 
\[\pPP{\left|\|Aw\|^2_2 -1\right| < \epsilon}{A \sim \D}  \geq 1- \delta.\]
\end{theorem}
\begin{proof}
Let $V$ be the random orthogonal matrix as defined above, and let $x =(x_1, \ldots, x_d)^t = Vw$. Then $Aw = \sqrt{s_0^{-1}} (x_1, \ldots, x_k)^t$, and
\[ ||Aw||_2^2 = \frac{1}{s_0} (x_1^2 + \cdots+ x_k^2).\]
Since $V$ is orthonormal and $||w||_2 =1$, we have $||x||_2 =1$, hence $x \in S^{d-1}$.  We observe that since $V$ is a random orthogonal matrix, for fixed $w\in S^{d-1}$, $x=Vw$ is a random variable, uniformly distributed on $S^{d-1}$.
Hence, 
\[\pPP{\left| \|Aw\|^2_2-1 \right|> \epsilon}{A\sim\D} = \pPP{ \, \vline \frac{1}{s_0}\sum_{i=1}^k x_i^2 -1 \vline > \epsilon}{x\sim S^{d-1}},\]
where $x\sim S^{d-1}$ means that $x$ is a random variable uniformly distributed on $S^{d-1}$.  Let $s= \sum_{i=1}^k x_i^2$. Then, $s \in [0,1]$ and the probability above becomes
\begin{equation}\label{eq:exteriorbounds}
\pPP{ s < s_0(1-\epsilon)}{x\sim S^{d-1}}
				+ \pPP{ s> s_0(1+\epsilon)}{x\sim S^{d-1}}\leq 2Ce^{-\frac{k-2}{4}\epsilon^2\left(1-\frac{2}{3}\epsilon\right)},\end{equation}
by assumption.  We observe that when 
\begin{equation}\label{eq:form:o1}k> 4\epsilon^{-2} \log \left(\frac{1}{\delta} \right)\left[1+ \frac{2\epsilon}{3-2\epsilon} + \frac{\log(2C)}{\log\left( \frac{1}{\delta}\right)} \cdot \frac{1}
			{1-2\epsilon/3 }+ \frac{2\epsilon^2}{4\log \left(\frac{1}{\delta}\right)}\right]   = 4\epsilon^{-2} \log \left(\frac{1}{\delta} \right)\left[1+o(1)\right] ,\end{equation}
the right-hand-side of Inequality (\ref{eq:exteriorbounds}) is less than $\delta$.
In this case, the $o(1)$ term needed in the theorem statement appears in Inequality (\ref{eq:form:o1}).  Therefore, when $k>  4\epsilon^{-2} \log \left(\frac{1}{\delta} \right)\left[1+o(1)\right] $, 
		the distribution $\D$  is an $(\epsilon, \delta)$-JL distribution.	
\end{proof}

\subsection{Lower Bound for Arbitrary Distributions} 

In this section, we prove an optimal lower bound on the limit in Theorem \ref{Thm:IntroThm} that matches the upper bound from the previous section.  The proof of this lower bound is the main challenge in this paper.  We begin with the following key lemma:

\begin{lemma} \label{lemma_c}
Let $x =(x_1, \ldots, x_d)^t$ be a random variable, uniformly distributed on $S^{d-1}$, $\Psi(x)=(s,u,v)$, and $s_0=k/d$.
Suppose that
\[   \min\{ \pP{s>s_0(1+\epsilon)}, \pP{s<s_0(1-\epsilon)}\} \geq L,\]
where $s$ is a random variable with probability distribution $Bf(s)$ on $[0,1]$.
For any function  $c(u,v)>0$ depending  only on $u \in S^{k-1}$ and $v \in S^{d-k-1}$ (i.e., independent of $s$), we have
\[
\pPP{|s c -1| > \epsilon}{x\sim S^{d-1}}  \geq   L.
\]
\end{lemma}
\begin{proof}
By the equality of differential forms in Equation (\ref{eq:equalitymeasures}), 
\begin{eqnarray*}
\pP{|sc -1| > \epsilon}  &  =&  \int_{|s c -1| > \epsilon} Bf(s) ds  \frac{d\Omega_{k-1}}{\VV{S^{k-1}}{k-1}} \frac{d\Omega_{d-k-1}}{\VV{S^{d-k-1}}{d-k-1}}\\
				& = & \int_{S^{k-1}\times S^{d-k-1}} \left( \int_{|s c -1| > \epsilon} Bf(s) ds \right)  \frac{d\Omega_{k-1}}{\VV{S^{k-1}}{k-1}} \frac{d\Omega_{d-k-1}}{\VV{S^{d-k-1}}{d-k-1}}.
\end{eqnarray*}
Our goal is to find a lower bound on the integral $\int_{|s c -1| > \epsilon} Bf(s) ds$.  Due to the independence of $u$, $v$, and $s$, $c(u,v)$ is a fixed positive constant within this integral.  We observe that $|sc-1|>\epsilon$ consists of two intervals, $s<(1-\epsilon)/c$ and $s>(1+\epsilon)/c$ and consider two cases depending on the value of $c$.

We begin by recalling that
$$\pP{s>s_0 (1+\epsilon)}  =  \int_{s>s_0 (1+\epsilon)} B f(s)ds\quad\text{and}\quad\pP{s<s_0 (1-\epsilon)}   =   \int_{s<s_0(1-\epsilon)} Bf(s)ds.$$
If  $c \geq  s_0$, then   $(1+\epsilon)/c \leq  (1+\epsilon)/s_0$, and, hence
\[  \int_{|s c -1| > \epsilon} Bf(s) ds  \geq \int_{s > (1+ \epsilon)/c} Bf(s) ds \geq \int_{s > (1+ \epsilon)/s_0} Bf(s) ds \geq L.\]
On the other hand, if  $c < s_0$, then  $(1-\epsilon)/s_0 < (1-\epsilon)/c$, then
\[  \int_{|s c -1| > \epsilon} Bf(s) ds  \geq \int_{s < (1- \epsilon)/c} Bf(s) ds \geq \int_{s < (1- \epsilon)/s_0} Bf(s) ds \geq L.\]
Therefore, the integral $\int_{|s c -1| > \epsilon} Bf(s) ds$ is bounded from below by $L$, and 
\[  \pP{|sc -1| > \epsilon}     \geq   
\int_{S^{k-1}\times S^{d-k-1}} L    \frac{d\Omega_{k-1}}{\VV{S^{k-1}}{k-1}} \frac{d\Omega_{d-k-1}}{\VV{S^{d-k-1}}{d-k-1}}  =  L.\vspace{-.37in}\]
\end{proof}
\vspace{.12in}

We now show that when $k\leq\eta\epsilon^{-2}\log(1/\delta)$ with $\eta<4$, and $\epsilon$ and $\delta$ are sufficiently small, there does not exist an $(\epsilon, \delta)$-JL distribution on  $\R^{k\times d}$.
This fact, combined with the results in Section \ref{Sec: OrthogProj}, shows that the limit 
appearing in Theorem \ref{Thm:IntroThm} exists and equals $1$.  It is challenging to show this directly; instead, we consider the following related problem:
By definition, for a probability
distribution $\D$ on  $\R^{k\times d}$
to be an  $(\epsilon, \delta)$-JL distribution, the following inequality must hold for every $w \in S^{d-1}$:
\[	\pPP{|\|Aw\|_2^2 -1| > \epsilon }{A \sim \D} < \delta.\]
Hence,
\begin{equation}\label{eq4.2.1}
\pPP{|\|Aw\|_2^2 -1| > \epsilon }{A \sim \D, \  w \sim S^{d-1}} < \delta,
\end{equation}
where $w \in S^{d-1}$ is a random variable distributed uniformly on $S^{d-1}$.  Our approach is to prove that, for every $A \in \R^{k \times d}$, 
\begin{equation}\label{eq:oppositedir}
\pPP{|\|Aw\|_2^2 -1| > \epsilon }{w \sim S^{d-1}} > \delta.
\end{equation}
When Inequality (\ref{eq:oppositedir}) holds for all $A$, then Inequality (\ref{eq4.2.1}) can not hold for any distribution $\D$ on $\R^{k \times d}$.  Therefore, an $(\epsilon, \delta)$-JL distribution does not exist.  We make this precise in the following theorem:

\begin{theorem} \label{thm_lowbound}
Suppose that $\eta<4$ and let $k(\epsilon,\delta)=\left\lfloor\eta\epsilon^{-2} \log \left(  \frac{1}{\delta}  \right)\right\rfloor$.
Let $s_0=k/d$, and suppose that, for every $\epsilon$, $\delta$, and $s_0$ sufficiently small (to make $s_0$ sufficiently small, $d$ must be sufficiently large),
\[    \min\{\pP{s>s_0(1+\epsilon)},\pP{s<s_0(1-\epsilon)}\} \geq C\delta^{\frac{\eta}{4}\gamma},\]
where $C>0$ is an absolute constant, and $\gamma$ approaches $1$ as $\epsilon$, $\delta$, and $s_0$ approach $0$.
Then, by decreasing $\epsilon$, $\delta$, and $s_0$ as needed, for every matrix   $A \in \R^{k(\epsilon,\delta) \times d}$, 
				\[\pPP{\left|\|Aw\|^2_2-1\right|> \epsilon}{w \sim S^{d-1}} > \delta.\]
\end{theorem}

\begin{proof}
We assume that $A$ has rank $k=k(\epsilon,\delta)$ since, if not, we may reduce $k$ (and decrease $\eta$ correspondingly) to the rank of $A$.
Let $A=U\Sigma V^t$ be the singular value decomposition of $A$ where $U$ is a $k \times k$ orthonormal matrix, $V= (v_1, \ldots, v_d)$ is a $d \times d$ orthonormal matrix, 
and $\Sigma$ is a $k \times d$ diagonal matrix with $\lambda_i> 0$ its entry at $(i,i)$ for $1 \leq i \leq k$.
Let   
\[ x = (x_1, \ldots, x_d)^t = V^t w.\] 
Since $V$ is orthonormal,  we have  $x \in S^{d-1}$.  We observe that since $w$ is a uniformly distributed random variable on $S^{d-1}$, $V^tw$ is also a uniformly distributed random variable on $S^{d-1}$.
Therefore, since $U$ is orthonormal, we have
      \[ \|Aw\|^2_2 = \|U \Sigma x\|^2_2 = \|\Sigma x\|^2_2 = \sum_{i=1}^k \lambda_i^2x_i^2.\]	
Let $\Psi(x)=(s,u,v)$ where $s = x_1^2 + \cdots + x_k^2$.  We restrict our attention to the case where $s\in(0,1)$ since the complement has zero measure.
Let 
\[ c =  \sum_{i=1}^k \lambda_i^2x_i^2/s = \|\Sigma u \|^2_2,\]
then
\[
\pPP{\left|\|Aw\|^2_2-1\right|> \epsilon}{w\sim S^{d-1}}  =  \pPP{\left|s c -1\right|>\epsilon}{x\sim S^{d-1}}.\]
Due to the independence of $u$, $v$, and $s$, it follows that $c$ depends only on $u$.  Therefore, by Lemma \ref{lemma_c}, it follows that 
$$
\pPP{\left| \|Aw\|_2^2 -1\right| > \epsilon}{w\sim S^{d-1}}  \geq  C\delta^{\frac{\eta}{4}\gamma}.
$$
It follows that for $\epsilon$, $\delta$, and $s_0$ sufficiently small, $C\delta^{\frac{\eta}{4}\gamma}>\delta$.
\end{proof}
Since $d$ grows as $s_0$ approaches $0$, it follows from Theorem \ref{thm_lowbound}, that for $d$ sufficiently large, there is no $(\epsilon, \delta)$-JL distribution when $k< \eta\epsilon^{-2} \log \left(  \frac{1}{\delta}  \right)$ for $\eta<4$.  Therefore, $k_0(\epsilon,\delta)>\eta\epsilon^{-2}\log\left(\frac{1}{\delta}\right)$.
We collect the results of Theorems \ref{Thm: exists} and \ref{thm_lowbound} in the following corollary:
\begin{cor}
Assume the hypotheses on 
$\pP{s>s_0(1+\epsilon)}$ and $\pP{s<s_0(1-\epsilon)}$
from Theorems \ref{Thm: exists} and \ref{thm_lowbound} hold.
\begin{enumerate}[(a)]
\item There exists an $o(1)$ function that approaches $0$ as $\epsilon$ and $\delta$ approach zero such that if $k>4\epsilon^{-2} \log \left( \frac{1}{\delta} \right)\left[1+o(1)\right]$, then there exists a JL distribution. 
\item If $k(\epsilon,\delta)=\left\lfloor\eta\epsilon^{-2} \log \left(  \frac{1}{\delta}  \right)\right\rfloor$, then, by decreasing $\epsilon$ and $\delta$, and increasing $d$, there is no $(\epsilon, \delta)$-JL distribution for any $k'\leq k(\epsilon,\delta)$.
\end{enumerate}
\end{cor}
This proves the main result in the paper.  In the following sections, we provide the more technical results that verify the assumptions in Theorems \ref{Thm: exists} and \ref{thm_lowbound}.

\section{Uniform Distributions on Unit Spheres in High Dimensions} \label{Sec:UniformMeasure}

In this section, we prove the explicit relationship between the surface area differential forms $d\Omega_{d-1}$, $d\Omega_{k-1}$, and $d\Omega_{d-k-1}$.  In particular, we prove that

\begin{theorem} \label{Thm_ud}
Under the almost bijective map $\Psi: S^{d-1}   \rightarrow    [0,1] \times S^{k-1} \times S^{d-k-1}$, we have equality of the surface area differential forms on $S^{d-1}$, $S^{k-1}$, and $S^{d-k-1}$, i.e.,
  \[ d\Omega_{d-1} = \frac{1}{2} f(s)ds\, d\Omega_{k-1} d\Omega_{d-k-1},\]
   where $f(s)=s^{(k-2)/2}(1-s)^{(d-k-2)/2}$.
Equivalently,  in terms of probability distribution measures,
 \[  \frac{d\Omega_{d-1}}{\VV{S^{d-1}}{d-1}} =  B f(s) ds \frac{d\Omega_{k-1}}{\VV{S^{k-1}}{k-1}} \frac{d\Omega_{d-k-1}}{\VV{S^{d-k-1}}{d-k-1}}.\]
Hence, the uniform distribution on $S^{d-1}$ can be identified with the product distribution on $[0,1] \times S^{k-1} \times S^{d-k-1}$ where the distribution of $s$ on $ [0,1]$ has density function $Bf(s)$ and 
 \begin{equation} \label{Eqn:B}
 		 B = \frac{1}{2} \cdot \frac{\VV{S^{k-1}}{k-1} \cdot \VV{S^{d-k-1}}{d-k-1}}{\VV{S^{k-1}}{k-1}} = 
  \frac{\Gamma(\frac{d}{2})}{\Gamma(\frac{k}{2}) \Gamma(\frac{d-k}{2})}.
  \end{equation}
 \end{theorem}
 This theorem is based on the following lemma, which is well-known to experts, but is included here for completeness.
 
\begin{lemma}  \label{Lemma:SAofSphere}
Let $x=(x_1, \ldots, x_d)$ with $x_d>0$ be a  point on the upper hemisphere of $S^{d-1}$. Then the surface area  measure of the unit sphere $S^{d-1}$ at $x$ is
				\[d\Omega_{d-1} = \frac{1}{x_d}dx_1 \dots dx_{d-1}.\]
\end{lemma}
Before we begin the proof, we recall the approach for $S^2$ in $3$-dimensional space.  We consider the upper hemisphere of $S^2$ as the graph of a function over $D^2$, where $D^{d-1}$ denotes
 $(d-1)$-dimensional disk, namely
\[ D^{d-1}=\left\lbrace \hat{x} \in \R^{d-1}: \sum_{i=1}^{d-1} \hat{x}_i^2 \leq 1\right\rbrace.\]
We then integrate over the disk $D^2$ to calculate the surface area of $S^2$.  In particular, the integrand is the limit of the ratios of the area of a square in $D^2$ to the area of the corresponding parallelogram above the square in the tangent space of $S^2$ as the square shrinks a point.  In the case of the sphere, the parallelogram's area is calculated using the cross product, but we must replace the use of the cross product in higher dimensions.
				
\begin{proof}
In $d$-dimensional space, we consider the upper hemisphere of $S^{d-1}$ as the graph of a function over the $(d-1)$-dimensional disk $D^{d-1}$.
We  construct a pair of $(d-1)$-dimensional parallelepipeds as follows: $P^{d-1}$ is in the tangent space of $D^{d-1}$ and $Q^{d-1}$ is in the tangent space of $S^{d-1}$.  Then, we take the limit of their $(d-1)$-dimensional volumes as $P^{d-1}$ approaches a point.  
			Due to complications in taking the $(d-1)$-dimensional volume in $d$-dimensional space, we extend both $P^{d-1}$ and $Q^{d-1}$ to associated, full-dimensional parallelepipeds.
			
			Let $(\hat{x}_1,\dots,\hat{x}_{d-1})\in D^{d-1}$ and define $\phi: D^{d-1} \rightarrow \R_{\geq 0}$ as 
			\[\phi(\hat{x}_1, \ldots, \hat{x}_{d-1}) = \sqrt{1-\sum_{i=1}^{d-1} \hat{x}_i^2}.\]
			We observe that the graph of this function is the upper hemisphere of $S^{d-1}$.  
			We now extend this map to $D^{d-1}\times\mathbb{R}$ as $\Phi_d: D^{d-1} \times \R \rightarrow \R^d$ defined by
				\[(\hat{x}_1, \ldots, \hat{x}_{d-1}, \hat{x}_d) \mapsto \left( (1+\hat{x}_d)\hat{x}_1, (1+\hat{x}_d)\hat{x}_2, \ldots, (1+\hat{x}_d)\hat{x}_{d-1}, (1+\hat{x}_d)\phi(\hat{x}_1, \ldots, \hat{x}_{d-1}) \right).\]  
				We observe that $\Phi_d|_{D^{d-1}\times\{0\}}$ maps the disk $D^{d-1}\times\{0\}$ surjectively onto the graph of $\phi$, i.e., the upper hemisphere of $S^{d-1}$, see Figure 
			\ref{Figure: MapPhi}.
			
				\begin{figure}[htb]
				\begin{center}
				\begin{tikzpicture}
				\begin{scope}
				\begin{axis}
				    [  
				        hide axis,
				        view={0}{10},
				        z buffer=sort,
				     height = 2.5in,
				     width = 2in
				    ]
				        \addplot3
				        [   domain=0:360,
				            y domain=0:90,
				            surf,
				            shader=flat,
				            blue,
				            opacity=0.4
				        ] ({sin(y)*cos(x)},{sin(y)*sin(x)},{cos(y)});
				        \addplot3
				        [ domain = 0:360,
				        y domain = 0:1,
				        surf,
				        shader=flat,
				        black	,
				        opacity =0.4	        
				        ] ({cos(x)*y},{sin(x)*y},{-2});
				        \draw [<-,line width=1pt] (-.4,0,-.5) --node[left]{$\Phi_d$} (-.4,0,-1.5);
				        \draw [->,line width=1pt] (.4,0,-.5) --node[right]{$\pi$} (.4,0,-1.5);
				    \end{axis}
				\end{scope}
				\end{tikzpicture}
				\end{center}
				\caption{$\Phi_d$ maps the disk $D^{d-1}\times\{0\}$ surjectively onto the upper hemisphere of the $(d-1)$-dimensional unit sphere $S^{d-1}$. 
				We observe 
				that $(\Phi_d)^{-1} = \pi$ is the projection map onto the first $d-1$ coordinates.
				 \label{Figure: MapPhi}}
				\end{figure}
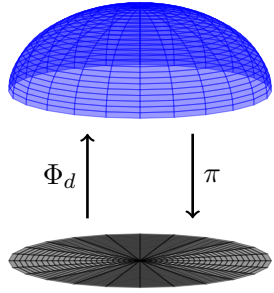

We recall that, at $(\hat{x}_1,\dots,\hat{x}_{d-1})\in D^{d-1}$, the tangent space is $\mathbb{R}^{d-1}$ and we define the parallelpiped $P^{d-1}$ in the tangent space by the vectors $\Delta \hat{x}_i e_i$ of length $\Delta \hat{x}_i$ in the direction of the $i^{\text{th}}$ standard basis vector $e_i$ of $\R^{d-1}$.  Similarly, the tangent space at $(\hat{x}_1,\dots,\hat{x}_{d-1},0)\in D^{d-1}\times\R$ is $\mathbb{R}^d$, and we define the parallelepiped $P^d$ in this tangent space by the vectors $\Delta \hat{x}_i e_i$ for $1\leq i\leq d-1$ and $he_d$ for the final direction.  Then, as the vector $he_d$ is perpendicular to the tangent vectors 
			of the disk $D^{d-1}$, the $d$-dimensional volume of $P^d$ can be computed in terms of the $(d-1)$-dimensional volume of $P^{d-1}$ and the height $h$, i.e.,
				\begin{equation*} 
				\VV{P^d}{d} = h \VV{P^{d-1}}{d-1}.
				\end{equation*}

			Next, we let $Q^{d-1}$ and $Q^d$ be the images of $P^{d-1}$ and $P^d$ under the Jacobian of $\Phi_d$, i.e., $\Jac\Phi_d$, respectively.  Note that the Jacobian of $\Phi_d$, when restricted to $D^{d-1} \times \{0\}$, is
			\begin{align*}
			\left(\Jac \Phi_d \right)|_{D^{d-1} \times \{0\}} & =  
					\left[
					\begin{array}{ccc|c}
					&&&\hat{x}_1 \\[.5cm]
					&I&&\vdots\\[.5cm]
					&&& \hat{x}_{d-1} \\[.3cm]\hline
					-\frac{ \hat{x}_1}{\phi(\hat{x}_1, \ldots, \hat{x}_{d-1})}&\dots&-\frac{ \hat{x}_{d-1}}{\phi(\hat{x}_1, \ldots, \hat{x}_{d-1})}& \phi(\hat{x}_1, \ldots, \hat{x}_{d-1})
					\end{array}
					\right].
				\end{align*}
				Since the Jacobian acts on tangent vectors, $Q^{d-1}$ is defined by the vectors 	
				\[\Delta \hat{x}_i \left(f_i - \frac{\hat{x}_i}{\phi(\hat{x}_1, \ldots, \hat{x}_{d-1})}f_d\right),\]
			for $1\leq i\leq d-1$, where the $f_i$ is the $i^{\text{th}}$ standard basis vector of $\mathbb{R}^d$.
			Moreover, $Q^d$ is defined by these vectors as well as the image of $he_d$, i.e., 
								\[h\left( \hat{x}_1, \ldots,\hat{x}_{d-1}, \phi(\hat{x}_1, \ldots, \hat{x}_{d-1} )\right) .\]
			We observe that the vectors $\Delta \hat{x}_i \left(f_i - \frac{\hat{x}_i}{\phi(\hat{x}_1, \ldots, \hat{x}_{d-1})}f_d\right)$ for $1\leq i \leq d-1$ are tangent vectors to the sphere $S^{d-1}$, 
			and 
			that $h\left( \hat{x}_1, \ldots,\hat{x}_{d-1}, \phi(\hat{x}_1, \ldots, \hat{x}_{d-1} )\right)$ is the outward pointing surface normal with length $h$.  Since the tangent vectors are perpendicular to the outward point normal, the volumes of $Q^{d-1}$ and $Q^d$ have a similar relationship as the volumes of $P^{d-1}$ and $P^d$, i.e., 
				\begin{equation*}
				\VV{Q^d}{d} = h \VV{Q^{d-1}}{d-1}.
				\end{equation*}
				Therefore, the ratio between the $d$-dimensional volumes of $Q^d$ and $P^d$ is the same as the ratio of the $(d-1)$-dimensional volumes of $Q^{d-1}$ and $P^{d-1}$.
					
			Since $Q^d$ is the image of $P^d$ under the linear map $\Jac \, \Phi_d \vline_{D^{d-1} \times \{0\}}$, it follows that 		 \begin{equation*} 
				\VV{Q^d}{d} = \, \left\vert \det \left(\Jac \Phi_d \right)\vline_{\left(\hat{x}_1, \ldots, \hat{x}_{d-1}, 0 \right)} \, \right\vert  \VV{P^d}{d}.
				\end{equation*}
				Therefore, the ratio of the volumes of $Q^{d-1}$ and $P^{d-1}$ is $\left| \det \left(\Jac \Phi_d \right)\vline_{\left(\hat{x}_1, \ldots, \hat{x}_{d-1}, 0 \right)} \, \right|$.
				It is straight-forward to compute the determinant of $\Jac(\Phi_d)$ at the point  $\left( \hat{x}_1, \ldots,\hat{x}_{d-1}, 0\right)$ via a few row reductions to eliminate the first $d-1$ entries in the last row and turn the matrix into an upper triangular matrix whose lower right corner is $\frac{1}{\phi(\hat{x}_1, \ldots, \hat{x}_{d-1})}$.
				Hence, the determinant of 
			$\Jac(\Phi_d)$ is $\frac{1}{\phi(\hat{x}_1, \ldots, \hat{x}_{d-1})}$, which is the desired scaling factor.

			Therefore, $\frac{1}{\phi(\hat{x}_1, \ldots, \hat{x}_{d-1})}$ is the local factor in the stretching of the surface area in the map from $D^{d-1}$ to $S^{d-1}$.  
			We recall that the coordinates $x_1,\dots,x_d$ are the coordinates on the upper hemisphere and $\hat{x}_1,\dots,\hat{x}_{d-1}$ are the coordinates on $D^{d-1}$.
			 Since, under the map $\Phi_d|_{D^{d-1}\times\{0\}}$, $x_i=\hat{x}_i$ for $1\leq i\leq d-1$, it follows that
			$$
			d\hat{x}_1\dots d\hat{x}_{d-1}=dx_1\dots dx_{d-1}\qquad\text{and}\qquad \phi(\hat{x}_1,\dots,\hat{x}_{d-1})=x_d.
			$$
			From here, the result follows directly.		
		\end{proof}


\noindent {\it Proof of Theorem \ref{Thm_ud}}. 
Let $(w_1, \ldots, w_d)$, $(x_1, \ldots,x_k)$, and $(y_1, \ldots, y_{d-k})$ be coordinates of points on the $(d-1)$-dimensional unit sphere, 
			the $(k-1)$-dimensional unit sphere, and the $(d-k-1)$-dimensional unit sphere, respectively. 
\rmv{			
To show $\Psi$ is onto, consider an arbitrary $w \in S^{d-1}$, and let $s = \sum_{i=1}^k w_i^2$. Then, $1-s = \sum_{i=k+1}^d w_i^2$. As a result, $w$ can be 
			decomposed into two unit vectors scaled by $\sqrt{s}$ and $\sqrt{1-s}$, respectively. That is, when $s \in (0,1)$,
				\[w = \sqrt{s}\left(\frac{w_1}{\sqrt{s}}, \ldots, \frac{w_k}{\sqrt{s}}\right) \times \sqrt{1-s}\left(\frac{w_{k+1}}{\sqrt{1-s}}, \ldots, \frac{w_d}{\sqrt{1-s}}\right).\]
			If $s =0$, $w= 0 (1,0, \ldots, 0) \times 1(w_{k+1}, \ldots, w_d)$. If $s=1$, then $w=1(w_1, \ldots, w_k) \times 0 (1,0, \ldots, 0)$.			
			Therefore, we let $x= \left(\frac{w_1}{\sqrt{s}}, \ldots, \frac{w_k}{\sqrt{s}}\right) $ and $y=\left(\frac{w_{k+1}}{\sqrt{1-s}}, \ldots, \frac{w_d}{\sqrt{1-s}}\right)$ when $s \in (0,1)$; 
			$x=e_1$ and $y = (w_{k+1}, \ldots, w_d)$ when $s=0$; and $x=(w_1, \ldots, w_k)$ and $y=e_1$ when $s=1$. We 
			observe that 
			$x \in S^{k-1}$ and $y \in S^{d-k-1}$ and that $\Psi(s,x,y)=w$.
}			
Let $(\hat{w}_1,  \ldots, \hat{w}_{d-1})$, $(\hat{x}_1, \ldots,\hat{x}_{k-1})$, and $(\hat{y}_1, \ldots, \hat{y}_{d-k-1})$ be the coordinates of points on the disks $D^{d-1}$, 
$D^{k-1}$ and $D^{d-k-1}$, respectively.	Let
				\[\varphi: [0,1] \times D^{k-1} \times D^{d-k-1} \rightarrow D^{d-1}\]
be defined by 
				\begin{align*}
				s \times (\hat{x}_1, \ldots, \hat{x}_{k-1}) \times & (\hat{y}_1, \ldots, \hat{y}_{d-k-1}) \mapsto \\
				 & \left( \sqrt{s}\hat{x}_1, \ldots, \sqrt{s}\hat{x}_{k-1},\sqrt{s} 
				\sqrt{1-\sum_{i=1}^{k-1} \hat{x}_i^2} ,  \sqrt{1-s} \hat{y}_1, \ldots, \sqrt{1-s} \hat{y}_{d-k-1} \right).
				\end{align*}
			We observe that $\varphi$ maps the disks $D^{k-1}$ and $D^{d-k-1}$ onto the half of the disk $D^{d-1}$ whose $k^{th}$ coordinate is nonnegative.
	As the measure of the image is the measure 
			of the preimage scaled by the determinant of the Jacobian of $\varphi$, the surface area measure of the disk $D^{d-1}$ is
				\begin{equation} \label{Eq: det_disks}
					d\hat{w}_1 \dots d\hat{w}_{d-1} = |\det \Jac( \varphi)	| ds\,d\hat{x}_1 \dots d\hat{x}_{k-1}  d\hat{y}_1 \dots d\hat{y}_{d-k-1}.
				\end{equation}	 
			The Jacobian of $\varphi$ for $s\in(0,1)$ is 		
				$$
				\text{Jac }\varphi = 
				\left[
				\begin{array}{c|ccc|ccc}
					\frac{\hat{x}_1}{2\sqrt{s}}	& \sqrt{s} & \dots & 0 & 0& \dots & 0 \\[.3cm]
					\vdots & \vdots & \ddots & \vdots & \vdots & \ddots & \vdots \\[.4cm]
					\frac{\hat{x}_{k-1}}{2\sqrt{s}} & 0 & \dots & \sqrt{s} &  0 & \dots & 0\\ [.3cm] \hline
					\frac{\sqrt{1-\sum_{i=1}^{k-1}\hat{x}_i^2}}{2\sqrt{s}} &  \frac{-\sqrt{s}\cdot \hat{x}_1}{\sqrt{1-\sum_{i=1}^{k-1}\hat{x}_i^2}} & \dots & \frac{-\sqrt{s}\cdot 
					\hat{x}_{k-1}}
					{\sqrt{1-\sum_{i=1}^{k-1}\hat{x}_i^2}} & 0 & \dots & 0 \\[.35cm] \hline
					\frac{-\hat{y}_1}{2\sqrt{1-s}} & 0 & \dots & 0 & \sqrt{1-s} & \dots & 0 \\[.3cm]
					\vdots & \vdots & \ddots & \vdots & \vdots & \ddots & \vdots \\ [.4cm]
					\frac{-\hat{y}_{d-k-1}}{2\sqrt{1-s}} & 0 & \dots & 0 & 0 & \dots & \sqrt{1-s}\\ [.3cm] 
				\end{array}
				\right].
				$$
			Eliminating all but the $k^{th}$ entry of the first column by
			adding multiples of the other columns to the first column, we obtain
				\[\det \Jac ( \varphi ) =	\frac{1}{2\hat{x}_k} s^{\frac{k-2}{2}} (1-s)^{\frac{d-k-1}{2}}.\]		
			Substituting this value into Expression (\ref{Eq: det_disks}), we have the surface area measure of $D^{d-1}$  in terms of the disks $D^{k-1}$ and $D^{d-k-1}$. 
			That is,	
				\begin{equation} \label{Eqn: SA of D}
				d\hat{w}_1 \dots d\hat{w}_{d-1} = \frac{1}{2\hat{x}_{k}} s^{(k-2)/2} (1-s)^{(d-k-1)/2} ds \, d\hat{x}_1 \dots d\hat{x}_{k-1} d\hat{y}_1 \dots 
				d\hat{y}_{d-k-1}. 
				\end{equation}
			We observe that the coordinates of the disk $D^{t-1}$ correspond to the first $t-1$ entries of coordinates  of the unit sphere $S^{t-1}$.  Therefore, we may extend $\varphi$ to the map $\Psi$, as defined above, where 
			$$\Psi^{-1}=\Phi_d\circ\varphi\circ (id_s\times(\Phi_k)^{-1}\times(\Phi_{d-k})^{-1}).
			$$			
			Employing 
			the 
			results of Lemma \ref{Lemma:SAofSphere} in various dimensions, we rewrite the surface measure of a unit sphere in terms of the surface measure of the corresponding disks:
				\begin{align*}
						 d\hat{x}_1 \dots d\hat{x}_{k-1} &= dx_1 \dots dx_{k-1} = x_k d\Omega_{k-1} \\
						d\hat{y}_1 \dots d\hat{y}_{d-k-1} &= dy_1 \dots dy_{d-k-1} = y_{d-k} d\Omega_{d-k-1} \\
					 d\hat{w}_1 \dots d\hat{w}_{d-1} &= dw_1 \dots dw_{d-1} = w_d d\Omega_{d-1}.  
				\end{align*}			
By applying the $\Psi$, we can substitute these three equalities into Equation (\ref{Eqn: SA of D}) to obtain
\begin{align*}
 d\Omega_{d-1} = \frac{1}{w_d} dw_1 \dots dw_{d-1}
       &= \frac{y_{d-k}}{2w_d} s^{(k-2)/2} (1-s)^{(d-k-1)/2} ds\, d\Omega_{k-1}  d\Omega_{d-k-1}\\
       &\hspace{-1in}= \frac{1}{2} s^{(k-2)/2} (1-s)^{(d-k-2)/2}  ds \, d\Omega_{k-1} d\Omega_{d-k-1}=\frac{1}{2} f(s)ds \, d\Omega_{k-1} d\Omega_{d-k-1},
\end{align*}
where the third equality follows from the fact that $w_d = \sqrt{1-s}y_{d-k}$ by the map $\Psi$.  Since the cases where $s=0$ or $s=1$ have measure $0$, the result follows. \qed


\section{Explicit Concentration Bounds}\label{Concentration}

Throughout  this section,  we assume that  $s$ is a random variable with probability distribution $Bf(s)$.  We define $s_0=\frac{k}{d}$, and further assume that $0\leq \epsilon,\delta\leq 1/2$, $k-4\geq\epsilon^{-2}$, and $s_0<0.4$.  We derive lower and upper  bounds for the following probabilities:
\[  \pP{s>s_0(1+\epsilon)}\quad\text{and}\quad \pP{s<s_0(1-\epsilon)},\]
using the probability density $Bf(s)$ for $s$ and $f(s)=s^{(k-2)/2}(1-s)^{(d-k-2)/2}$.
These bounds are instances of explicit  concentration theorems (or explicit  laws of large numbers) from probability theory.
Our goal is to formulate these bounds as precisely as possible so that the lower and upper bounds are asymptotically the same
when $\epsilon$ and $\delta$ approach $0$.
			
\subsection{Bounds for B} 

We recall that $\Gamma(1/2) = \sqrt{\pi}$, $\Gamma(1) = 1$, and $\Gamma(1+z) = z\Gamma(z)$, hence 
\[ \Gamma\left(\frac{d}{2}\right) =  \left(\frac{d}{2} -1\right)!  \mbox{  if $d$ is even},  \quad \quad  \Gamma\left(\frac{d}{2}\right) =  \left(\frac{d}{2} -1\right)\cdot \left(\frac{d}{2} -2\right) \cdots \frac{3}{2}\cdot \frac{1}{2} \cdot \sqrt{\pi}  \mbox{  if $d$ is odd}.\]
In this section, we derive lower and upper bounds for $B$, see Equation (\ref{Eqn:B}), by  using  the following form of Stirling's approximation of $n!$ due to Robbins \cite{Robbins}: 
					\[\sqrt{2\pi}n^{n+1/2}e^{-n}e^{\frac{1}{12n+1}} <  \Gamma(n+1) = n! < \sqrt{2\pi} n^{n+1/2}e^{-n}e^{\frac{1}{12n}}.\]
Since we are interested in the asymptotic behavior, we focus on the case where $d$ is even.  This choice does not affect the asymptotic results of our paper, but the calculations are more straight-forward in this case.  We leave the details for the case where $d$ is odd to the interested reader.

\begin{lemma} \label{Lemma:B_Bound} 	Suppose $k$ and $d$ are both even. Then we have the following inequality\footnote{It is possible, to derive tighter bounds on constants, but the ones appearing here are sufficient for our proofs.  We leave the details of the tighter bounds to the interested reader}:
\[\frac{e^{-2}}{2\sqrt{\pi}} \cdot \frac{\left(d-2\right)^{(d-1)/2}}{ \left(k-2 \right)^{(k-1)/2}\left(d-k-2\right)^{(d-k-1)/2}} \leq B \leq \frac{e^{-1}}{2\sqrt{\pi}} \cdot 
		 \frac{\left(d-2\right)^{(d-1)/2}} {\left(k-2 \right)^{(k-1)/2}\left(d-k-2\right)^{(d-k-1)/2}}. \]
\end{lemma}			 			
\begin{proof}
				Using the bound on $n!$ from Robbins \cite{Robbins}, we obtain
				\[C_0 \frac{\left(d-2\right)^{(d-1)/2}}{ \left(k-2 \right)^{(k-1)/2}\left(d-k-2\right)^{(d-k-1)/2}} \leq B \leq C_1 \frac{\left(d-2\right)^{(d-1)/2}} {\left(k-2 
				\right)^{(k-1)/2}\left(d-k-2\right)^{(d-k-1)/2}}, \]
				where 
					\[C_0 = \frac{1}{2\sqrt{\pi}}e^{-1}e^{\frac{1}{6(d-2)+1}}e^{\frac{-1}{6(k-2)}}e^{\frac{-1}{6(d-k-2)}} \geq \frac{e^{-2}}{2\sqrt{\pi}}, \quad 
					\text{and}\]
					\[C_1 = \frac{1}{2\sqrt{\pi}}e^{-1}e^{\frac{1}{6(d-2)}}e^{\frac{-1}{6(k-2)+1}}e^{\frac{-1}{6(d-k-2)+1}} \leq \frac{e^{-1}}{2\sqrt{\pi}}.\vspace{-.25in}\]
				\end{proof}
				
			\begin{cor} \label{Corollary:Bsf_Bound}
				With $s_0 = k/d$, we have
					\[\frac{e^{-2}}{2\sqrt{\pi}}  \sqrt{k} \leq Bs_0f(s_0) \leq \frac{9e^{-1}}{\sqrt{2\pi}} \sqrt{k}.\]
				\end{cor} 
				\begin{proof}
					By evaluating $f$ at $s_0$ and replacing $B$ by its lower bound found in Lemma \ref{Lemma:B_Bound}, we obtain the lower bound
					\begin{align*}
					Bs_0f(s_0) \geq  \frac{e^{-2}}{2\sqrt{\pi}} \sqrt{k} \left(\frac{d-2}{d-k}\right)^{\frac{1}{2}} \left( \frac{k(d-2)}{d(k-2)} \right)^{\frac{k-1}{2}} \left( \frac{(d-k)(d-2)}{d(d-
					k-2)} \right)^{\frac{d-k-1}{2}}\geq  \frac{e^{-2}}{2\sqrt{\pi}} \sqrt{k} .
					\end{align*}
					Similarly, by using the upper bound in Lemma \ref{Lemma:B_Bound}, we obtain the upper bound
						$$
						Bs_0f(s_0) \leq \frac{e^{-1}}{2\sqrt{\pi}} \sqrt{k} \left(\frac{d-2}{d}\right)^{(d-1)/2} \left(\frac{k}{k-2}\right)^{(k-1)/2} \left( \frac{d-
						k}{d-k-2} \right)^{(d-k-1)/2} \frac{\sqrt{d}}{\sqrt{d-k}}
										\leq \frac{9e^{-1}}{\sqrt{2\pi}}  \sqrt{k},
						$$
				where the last inequality follows from $\frac{d}{d-k} \leq 2$ since $s_0<0.4$, and $\left(\frac{x}{x-2}\right)^{\frac{x-1}{2}} \leq 3$ for $x \geq 3$. 		
				\end{proof}
					

\subsection{Bounds on $\pP{s>s_0(1+\epsilon)}$ }
We begin by mentioning the following inequalities which are used in our arguments below:
						\begin{align}
								& \log (1+x) \geq x-\frac{x^2}{2} \quad \text{for $0<x<1$},  \label{Bound:Ln_plus} \\	
								& \log (1-x) \geq -x -x^2 \quad \text{for $0 <x<0.68$,} \label{Bound:Ln_minus2}\\
								& \log(1+x) \leq x \quad \text{for $x>-1$},  \quad \text{and} \label{Bound:Ln_plus_upper}\\
								& \log(1+x) \leq x - \frac{x^2}{2} + \frac{x^3}{3} \quad \text{for $x>-1$}. \label{Bound:Ln_plus_upper2} 
						\end{align}
These bounds can be verified by employing basic calculus techniques (e.g., derivatives and Taylor expansions) as well as sufficiently accurate approximations.  Using these inequalities, we derive the following bounds:
		
		 
				\begin{lemma}	\label{Lemma:lower+} 
					\[ \pP{s>s_0(1+\epsilon)} \geq \frac{e^{-2}}{4 \sqrt{\pi}} e^{-\frac{1}{4}(\sqrt{k}\epsilon+1)^2 \frac{1+s_0}{1-s_0}}.\]   
					Moreover, when $k<\eta\epsilon^{-2}\log\frac{1}{\delta}$,
					\[ \pP{s>s_0(1+\epsilon)}\geq \frac{e^{-2}}{4\pi}\delta^{\frac{\eta}{4}\gamma_1},\]
					where
					$$
					\gamma_1=\left(1+(\eta\log(1/\delta))^{-1/2}\right)^2\left(\frac{1+s_0}{1-s_0}\right).
					$$
					Additionally, $\gamma_1$ approaches $1$ as $\epsilon$, $\delta$, and $s_0$ approach $0$.
				\end{lemma} 
				
				\begin{proof}
Note that 
					\begin{equation} \label{Eqn: plus B int}
						\pP{s>s_0(1+\epsilon)} = B \int_{s>s_0(1+\epsilon)} f(s)ds=Bs_0 \int_\epsilon^{\frac{1}{s_0}-1} f(s_0(1+x)) dx, 
					\end{equation}
					via the substitution $s=s_0(1+x)$.

					Let $g(s) = s^{k/2}(1-s)^{(d-k)/2}$, then $f(s_0(1+x))$ can be expressed in terms of $g(s)$, namely,
					\begin{equation} \label{Eqn: f=g}
						f(s_0(1+x)) = \frac{g(s_0(1+x))}{s_0(1+x)\left(1-s_0(1+x)\right)}.
					\end{equation}
					To find a lower bound on $\pP{s>s_0(1+\epsilon)}$, we compute a bound on $g(s_0(1+x))$ from below. Taking the logarithm of $g(s_0(1+x))$, we find
					\begin{equation} \label{Eqn: log g}
						\log \left( g(s_0(1+x)) \right) = \log g(s_0) + \frac{d}{2}\left( s_0 \log (1+x) + (1-s_0)\log\left(1-\frac{s_0}{1-s_0}x\right)\right).
					\end{equation}		
					Restricting $x$ to the interval $0\leq x<1$, it then follows that $0< \frac{s_0}{1-s_0}x<0.68$ from the assumption that $s_0<0.4$.  We now bound the second term in Equation (\ref{Eqn: log g}) using Inequalities (\ref{Bound:Ln_plus}) and (\ref{Bound:Ln_minus2}), as follows:	
					\begin{align}
					 s_0 \log (1+x) + & (1-s_0)\log \left(1-\frac{s_0}{1-s_0}x\right)  
					 							 \geq   -\left(\frac{s_0(1+s_0)}{2(1-s_0)}\right)x^2. \label{Ineq: innerstatement}
					\end{align}
					Hence, by substituting Inequality (\ref{Ineq: innerstatement}) into Equation (\ref{Eqn: log g}) and exponentiating, we obtain the following lower bound for $g(s_0(1+x))$:
					\begin{equation} \label{Ineq: g term}
					g(s_0(1+x)) \geq g(s_0) e^{-\frac{k}{4}x^2\frac{1+s_0}{1-s_0} }.
					\end{equation}
					We also observe that since $s_0<0.4$ and $0\leq x<1$, the denominator of Equation (\ref{Eqn: f=g}) is bounded from below as follows:
					\begin{equation} \label{Ineq: remaining term}
						 \frac{1}{s_0(1+x)(1-s_0(1+x))} \geq \frac{1}{2s_0(1-s_0)}.
						 \end{equation}
					Therefore, by substituting Inequalities (\ref{Ineq: g term}) and (\ref{Ineq: remaining term}) into Equation (\ref{Eqn: f=g}) when $0\leq x<1$, we have
					\begin{equation}\label{ineq:f-lowerbound}
					f(s_0(1+x)) \geq  \frac{g(s_0)}{2s_0(1-s_0)} e^{-\frac{k}
						{4}x^2\frac{1+s_0}{1-s_0}} = \frac{1}{2}f(s_0)e^{-\frac{k}{4}x^2\frac{1+s_0}{1-s_0}}.
						\end{equation}

					Since $s_0<0.4$, we observe that $\frac{1}{s_0}-1=\frac{d-k}{k}>1.5$.  Since we assumed that $\epsilon<\frac{1}{2}$ and $k\geq 4+\epsilon^{-2}>4$, it follows that $\epsilon+k^{-1/2}<1$ and so $\epsilon+k^{-1/2}<1<\frac{1}{s_0}-1$.  Therefore, we further restrict $x$ to the interval $(\epsilon,\epsilon+k^{-1/2})$ and observe that Inequality (\ref{ineq:f-lowerbound}) applies in this range.  Therefore,
						\begin{equation}\label{ineq:lower:restrictingrange}
						Bs_0 \int_\epsilon^{\frac{1}{s_0}-1} f(s_0(1+x)) dx  \geq Bs_0 \int_\epsilon^{\epsilon + k^{-1/2}}f(s_0(1+x)) dx\geq \frac{1}{2} Bs_0 f(s_0)  \int_\epsilon^{\epsilon + k^{-1/2}}e^{-\frac{k}{4}x^2\frac{1+s_0}{1-s_0}} dx.
						\end{equation}
					Replacing $Bs_0f(s_0)$ with its lower bound given in Corollary \ref{Corollary:Bsf_Bound} and observing that the integrand is decreasing over an interval of width $k^{-1/2}$, Inequality (\ref{ineq:lower:restrictingrange}) is bounded from below by
					\begin{align*}
						\frac{1}{2} Bs_0 f(s_0)  \int_\epsilon^{\epsilon + k^{-1/2}}e^{-\frac{k}{4}x^2\frac{1+s_0}{1-s_0}} dx 
							& \geq \frac{e^{-2}}{4\sqrt{\pi}}  e^{-\frac{1}{4}(\sqrt{k}\epsilon+1)^2\frac{1+s_0}{1-s_0}} , 
					\end{align*}
					which completes the first inequality.  The second inequality follows by replacing $k$ by the given upper bound and simplifying.		
				\end{proof}		
						
\begin{lemma} \label{Lemma:Upper+}
\[  \pP{s>s_0(1+\epsilon)} \leq  \frac{27e^{-1}}{\sqrt{2\pi}}e^{-\frac{k-2}{4}\epsilon^2(1-\frac{2}{3}\epsilon) } .\]
\end{lemma}			
				
				\begin{proof}
				To derive an upper bound, we start with the expression for $\pP{s>(1+\epsilon)s_0}$ from Equation (\ref{Eqn: plus B int}).
				We first find an upper bound on $f(s_0(1+x))$.  
				We bound $f(s_0(1+x))$ using the inequality $1-x\leq e^{-x}$ for all $x$				
				as follows:		
					\begin{equation}\label{eqn:f:bound:first}
						f(s_0(1+x))  = f(s_0) (1+x)^{\frac{k-2}{2}} \left( 1- \frac{s_0}{1-s_0}x\right)^{\frac{d-k-2}{2}}   
									 \leq f(s_0) (1+x)^{\frac{k-2}{2}} \left( e^{\frac{-s_0}{1-s_0}x}\right)^{\frac{d-k-2}{2}},
								\end{equation}
									 Moreover, since $s_0<0.4$, 
									 \begin{equation}\label{eqn:exponent:bound}
									 \frac{s_0}{1-s_0}\frac{d-k-2}{2}=\frac{k}{d-k}\frac{d-k-2}{2}>\frac{k-2}{2}.
									 \end{equation}	
									 By applying Inequality (\ref{eqn:exponent:bound}) to Inequality (\ref{eqn:f:bound:first}), 
									 we derive the upper bound
									\begin{equation*}
									f(s_0) (1+x)^{\frac{k-2}{2}} e^{-\frac{k-2}{2}x}.
						\end{equation*}
						Therefore, by extending the region of integration in Inequality (\ref{Eqn: plus B int}), we find the following upper bound on the probability:
						\begin{equation}\label{eq:upperbound:firstpart}
						Bs_0\int_\epsilon^{\frac{1}{s_0}-1}f(s_0(1+x))dx\leq Bs_0f(s_0)\int_\epsilon^\infty(1+x)^{\frac{k-2}{2}} e^{-\frac{k-2}{2}x}dx.
						\end{equation}
				By integrating by parts, we observe that for any $\ell$ and $m$ with $1\leq\ell\leq m$,	
				\begin{align}
						\int_{\epsilon}^\infty (1+x)^{\ell} e^{-mx} dx  \leq \frac{1}{m}(1+\epsilon)^\ell e^{-m\epsilon} + \int_\epsilon^\infty (1+x)^{\ell-1} e^{-mx}dx. 
						\label{Equation:ineq}
					\end{align}
					Applying Inequality (\ref{Equation:ineq}) $\frac{k-2}{2}$ times to the integral in Inequality (\ref{eq:upperbound:firstpart}) and bounding the resulting geometric series from above  gives	
						\begin{align*}
							\int_\epsilon^\infty (1+x)^{\frac{k-2}{2}} e^{-\frac{k-2}{2}x} dx  \leq \frac{2e^{-\frac{k-2}{2}\epsilon}}{k-2} \left( (1+ \epsilon)^{\frac{k-2}{2}} + 
							  \dots + (1+\epsilon)^0 \right)  \leq \frac{2(1+\epsilon)}{\epsilon(k-2)} (1+\epsilon)^{\frac{k-2}{2}} e^{-\frac{k-2}{2}\epsilon}.  
						\end{align*}	
						By applying Inequality (\ref{Bound:Ln_plus_upper2}) to $(1+\epsilon)^{\frac{k-2}{2}} = e^{\frac{k-2}{2}\log (1+\epsilon)}$, we obtain the upper bound
						\[\frac{2(1+\epsilon)}{\epsilon(k-2)} (1+\epsilon)^{\frac{k-2}{2}} e^{-\frac{k-2}{2}\epsilon} \leq \frac{2(1+\epsilon)}{\epsilon(k-2)} e^{-\frac{k-2}
						{4}\epsilon^2(1-\frac{2}{3}\epsilon)}.\]
						Since $k-4\geq \epsilon^{-2}$, it follows that $\frac{(k-2)^2}{k}\geq k-4\geq\epsilon^{-2}$, and, hence, that $\epsilon(k-2)\geq \sqrt{k}$.  Therefore, we can further simplify our bound to 
						\begin{equation}
							 \frac{2(1+\epsilon)}{\epsilon(k-2)} e^{-\frac{k-2}{4}\epsilon^2(1-\frac{2}{3}\epsilon)} \leq 
							\frac{2(1+\epsilon)}{\sqrt{k}} e^{-\frac{k-2}{4}\epsilon^2(1-\frac{2}{3}\epsilon)}. 
							\label{Ineq: int_(1+x)}
						\end{equation} 
						
						By combining Inequalities (\ref{eq:upperbound:firstpart}) and (\ref{Ineq: int_(1+x)}), we find an upper bound on the probability as follows:
						$$
						Bs_0\int_\epsilon^{\frac{1}{s_0}-1}f(s_0(1+x))dx\leq B s_0 f(s_0)   \frac{2(1+\epsilon)}{\sqrt{k}} e^{-\frac{k-2}{4}\epsilon^2(1-\frac{2}{3}\epsilon)}.
						$$
	By applying the upper bound on $Bs_0f(s_0)$ from Corollary \ref{Corollary:Bsf_Bound} and the assumption that $\epsilon\leq\frac{1}{2}$, which completes the inequality.
				\end{proof}

	
\subsection{Bounds on $\pP{s<s_0(1-\epsilon)}$ }
		
We begin this section by including two additional inequalities on $\log(1-x)$.
		\begin{eqnarray}
		& & \log (1-x) \geq -x -\frac{x^2}{2}-x^3 \quad \text{for $0<x<0.815$},  \label{Bound:Ln_minus} \\
		&& \log(1-x)\leq -x-\frac{x^2}{2}\quad\text{ for }0\leq x<1.		 \label{Bound:Ln:minus:upper}
		\end{eqnarray}
These bounds can be justified using a similar approach as for Inequalities (\ref{Bound:Ln_plus}-\ref{Bound:Ln_plus_upper2}).  Using these inequalities, we derive the following bounds:
		
\begin{lemma} 
\[ \pP{s<s_0(1-\epsilon)} \geq \frac{e^{-2}}{2\sqrt{\pi}} e^{	- \frac{1}{4}\left(\frac{(\sqrt{k} \epsilon + 1)^2}{1-s_0} + 2(\sqrt[3]{k} \epsilon + k^{-1/6})^3\right)}.\]
Moreover, when $k<\eta\epsilon^{-2}\log\frac{1}{\delta}$,
$$
\pP{s<s_0(1-\epsilon)} \geq \frac{e^{-2}}{2\sqrt{\pi}} \delta^{\frac{\eta}{4}\gamma_2},$$
where
$$
\gamma_2=\frac{1}{1-s_0}\left(1+\frac{1}{\sqrt{\eta\log\frac{1}{\delta}}}\right)^2+2\left(\epsilon^{1/3}+\frac{1}{\sqrt[3]{\eta\log\frac{1}{\delta}}}\right)^3.
$$
Additionally, $\gamma_2$ approaches $1$ as $\epsilon$, $\delta$, and $s_0$ approach $0$.
\end{lemma}
				
\begin{proof}
The proof of this lemma is very similar to the proof of Lemma \ref{Lemma:lower+}, so we focus on the new details.  
The probability can be rewritten, using the substitution $s=s_0(1-x)$, as
					\begin{equation} \label{Eqn:minus int}
						\pP{s<(1-\epsilon)s_0} =  B\int_{s<s_0(1-\epsilon)} f(s)ds=Bs_0\int_{\epsilon}^1 f(s_0(1-x))dx.
						\end{equation}
						Using $g(s)$ as in Lemma \ref{Lemma:lower+}, it follows that
						\begin{equation}\label{Eq: f=g minus}
						f(s_0(1-x))=\frac{g(s_0(1-x))}{s_0(1-x)(1-s_0(1-x))}
						\end{equation}
						 and
						\begin{equation} \label{Eqn: g lowerbound}
							\log g(s_0(1-x)) = \log g(s_0) + \frac{d}{2}\left[ s_0 \log (1-x) + (1-s_0) \log \left( 1+ \frac{s_0}{1-s_0}x \right) \right].
						\end{equation}
					Since $0< x\leq 1$, it follows that $0<\frac{s_0}{1-s_0}x<0.68$ from the assumption that $s_0<0.4$.  Therefore, we can bound the second term in Equation (\ref{Eqn: g lowerbound}) using Inequalities (\ref{Bound:Ln_plus}) and (\ref{Bound:Ln_minus}), as follows:
					\begin{equation} 
							s_0 \log (1-x) +  (1-s_0) \log \left( 1+ \frac{s_0}{1-s_0}x \right) 
												 \geq \frac{-s_0}{2(1-s_0)}x^2 -s_0x^3. \label{Ineq: taylor}
						\end{equation}
											Substituting Inequality (\ref{Ineq: taylor}) into Expression (\ref{Eqn: g lowerbound}) and exponentiating, we get
						\begin{equation} \label{Ineq: g_lower}
							g(s_0(1-x)) \geq g(s_0)e^{ - \frac{k}{4}\left[\frac{x^2}{1-s_0} + 2x^3 \right]}.
						\end{equation}
						Since $s_0<0.4$, it follows that $\frac{s_0}{1-s_0}<1$ and, hence, that $(1-x)\left(1+\frac{s_0}{1-s_0}x\right)<1$.  Therefore, the denominator in Equation (\ref{Eq: f=g minus}) can be bounded from below by 
												\begin{equation} \label{Ineq: s0_2}
						  \frac{1}{s_0(1-x)\left(1-s_0(1-x)\right)} = \frac{1}{s_0(1-s_0)} \cdot \frac{1}{(1-x)\left(1+\frac{s_0x}{1-s_0} \right)} \geq \frac{1}{s_0(1-s_0)}.
						 \end{equation}
						 Therefore, by substituting Inequalities (\ref{Ineq: g_lower}) and (\ref{Ineq: s0_2}) into Expression (\ref{Eq: f=g minus}), when $\epsilon<x<1$, we have
						 						\begin{equation} \label{Ineq: fs_0_lower f_s0}
						f(s_0(1-x)) \geq f(s_0)e^{	- \frac{k}{4}\left(\frac{x^2}{1-s_0} + 2x^3 \right)}. 
						\end{equation}
												
						Since $\epsilon<\frac{1}{2}$ and $k\geq 4+\epsilon^{-2}>4$, it follows that $\epsilon+k^{-1/2}<1$.  Therefore, we restrict $x$ to the interval $(\epsilon,\epsilon+k^{-1/2})$, and observe that Inequality (\ref{Ineq: fs_0_lower f_s0}) applies in this range.  Therefore,
						\begin{equation}\label{ineq:lower:restrict:second}
							Bs_0  \int_{\epsilon}^1 f(s_0(1-x))dx \geq Bs_0 f(s_0) \int_\epsilon^{\epsilon+k^{-1/2}} e^{ - \frac{k}{4}\left(\frac{x^2}{1-s_0} + 2x^3 \right)} 																										dx.
							\end{equation}
											
					Replacing $Bs_0f(s_0)$ with its lower bound given in Corollary \ref{Corollary:Bsf_Bound} and observing that the integrand is decreasing on an interval of width $k^{-1/2}$, Inequality (\ref{ineq:lower:restrict:second}) is bounded from below by 
							$$
								Bs_0 f(s_0)\int_\epsilon^{\epsilon+k^{-1/2}} e^{ - \frac{k}{4}\left(\frac{x^2}{1-s_0} + 2x^3 \right)} dx  \geq   \frac{e^{-2}}{2\sqrt{\pi}} e^{	- 
								\frac{1}{4}\left(\frac{(\sqrt{k} \epsilon + 1)^2}{1-s_0} + 2(\sqrt[3]{k} \epsilon + k^{-1/6})^3\right)},
							$$
					which completes the first inequality.  The second inequality follows by replacing $k\geq 1$ by the given upper bound and simplifying.					
			\end{proof}

\begin{lemma} 
	\[  \pP{s<s_0(1-\epsilon)} \leq \frac{18\sqrt{2}e^{1/2}}{\sqrt{\pi}} e^{-\left(\frac{k}{4}\right)\epsilon^2}\leq \frac{18\sqrt{2}e^{1/2}}{\sqrt{\pi}} e^{-\left(\frac{k-2}{4}\right)\epsilon^2\left(1-\frac{2}{3}\epsilon\right)}.  \]
\end{lemma}
				
		\begin{proof}
				The proof of this lemma is very similar to the proof of Lemma \ref{Lemma:Upper+}, so we focus on the new details.
				To prove an upper bound, we start with the bound on $\pP{s<s_0(1-\epsilon)}$ from Equation (\ref{Eqn:minus int}).
				We first observe that 
				\begin{equation}\label{eq:fs:minus}
							f(s_0(1-x))  = f(s_0) (1-x)^{\frac{k-2}{2}} \left( 1+ \frac{s_0}{1-s_0}x\right)^{\frac{d-k-2}{2}}.
						\end{equation}
We now bound the logarithm of the second and third factors in Equation (\ref{eq:fs:minus}) using Inequalities (\ref{Bound:Ln_plus_upper}) and (\ref{Bound:Ln:minus:upper}) as follows:
						\begin{align} 
						\log \left( (1-x)^{\frac{k-2}{2}}  \left( 1+ \frac{s_0}{1-s_0}x\right)^{\frac{d-k-2}{2}} \right)
						& \leq \frac{k-2}{2} \left(-x-\frac{x^2}{2}\right) + \frac{d-k-2}{2} \left( \frac{s_0}{1-s_0}x \right)\label{Ineq:First:Half}
						\end{align}
						Since $\frac{s_0}{1-s_0}=\frac{k}{d-k}$ and $\epsilon<x<1$, Inequality (\ref{Ineq:First:Half}) further simplifies to 				
						$$
						\frac{k-2}{2} \left(-x-\frac{x^2}{2}\right) + \frac{d-k-2}{2} \left( \frac{k}{d-k}x \right) \leq  x -\left(\frac{k-2}{4}\right)x^2  \leq \frac{3}{2} - \left(\frac{k}{4}\right)x^2.
						$$
						Hence, for $\epsilon \leq x < 1$, we have
						\begin{equation} \label{Ineq:fs0_lower}
							 f(s_0(1-x)) \leq f(s_0) e^{3/2} e^{-\left(\frac{k}{4}\right)x^2} \leq  f(s_0) e^{3/2} e^{-\left(\frac{k}{4}\right)\epsilon x }.
						\end{equation}	
						
						Substituting Inequality (\ref{Ineq:fs0_lower}) into the integral of Equation (\ref{Eqn:minus int}), we find
						\begin{align}
						Bs_0\int_\epsilon^1f(s_0(1-x))dx
						&\leq Bs_0f(s_0)e^{3/2}\int_\epsilon^1 e^{-\left(\frac{k}{4}\right)\epsilon x}dx\notag\\
						&\leq Bs_0f(s_0)e^{3/2}\int_\epsilon^\infty e^{-\left(\frac{k}{4}\right)\epsilon x}dx
						\leq Bs_0f(s_0) \frac{4e^{3/2}}{k\epsilon}  e^{-\frac{k}{4}\epsilon^2}.\label{Ineq:Almost}
						\end{align}
						Since $k\geq\epsilon^{-2}$, Inequality (\ref{Ineq:Almost}) can be further simplified to
						$$
						Bs_0f(s_0) \frac{4e^{3/2}}{\sqrt{k}} e^{-\frac{k}{4}\epsilon^2}
						$$
						Applying the upper bound on $Bs_0f(s_0)$ from Corollary 
					\ref{Corollary:Bsf_Bound} completes the first inequality of the proof.  
					The final inequality follows from the fact that $k\geq (k-2)\left(1-\frac{2}{3}\epsilon\right)$.				
				\end{proof}
	
This completes the proof all of the conditions in Section \ref{Sec:Main:New}, and, therefore, completes the proof of our main theorem, Theorem \ref{Thm:IntroThm}.

\bibliographystyle{plain}
\bibliography{JLT_PaperBibliography}	

\end{document}